\documentclass[aps,pra,twocolumn,superscriptaddress,notitlepage,noshowkeys,10pt]{revtex4-2}

\usepackage{blindtext,graphicx,amsmath,amsthm,physics,mathtools,float,comment,xcolor,amssymb,bbm,stackengine,caption,subcaption,float,adjustbox}

\usepackage{babel}         % for different langages
\usepackage{blindtext}
\usepackage{graphicx}      % for inserting images
\usepackage{amsmath}
\usepackage{physics}       % for bra and ket
\usepackage{mathtools}     % for prescripts
\usepackage{comment} 
\usepackage{xcolor}
\usepackage{amssymb}       % for sets like the set of complex numbers
\usepackage{bbm}           % for Identity Matrix
\usepackage{stackengine}   % for double dash
\usepackage{caption}
\usepackage{subcaption}    % for subfigures
\usepackage{hyperref}

\usepackage{tikz}
\usetikzlibrary{shapes,backgrounds}

\mathtoolsset{showonlyrefs=true} % only number equations that are referenced

\stackMath                 % for double dash

\newcommand{\inv}{^{-1}}
\newcommand\tsup[2][2]{%
 \def\useanchorwidth{T}%
  \ifnum#1>1%
    \stackon[-.5pt]{\tsup[\numexpr#1-1\relax]{#2}}{\scriptscriptstyle\sim}%
  \else%
    \stackon[.5pt]{#2}{\scriptscriptstyle\sim}%
  \fi%
}
\newtheorem{thm}{Theorem}
\newtheorem{proposition}[thm]{Proposition}
\newtheorem{lemma}[thm]{Lemma}

\begin{document}

\title{
Long-term stability of driven quantum systems and the time-dependent Bloch equation}

\date{\today}
\author{Zsolt Szab\'o}
\email{zsolt.szabo@students.mq.edu.au}
\affiliation{School of Mathematical and Physical Sciences, Macquarie University, NSW 2109, Australia}
\affiliation{ARC Centre of Excellence for Engineered Quantum Systems, Macquarie University, NSW 2109, Australia}
\author{Kazuya Yuasa}
\affiliation{Department of Physics, Waseda University, Tokyo 169-8555, Japan}
\author{Daniel Burgarth}
\affiliation{Department Physik, Friedrich-Alexander-Universit\"at Erlangen-N\"urnberg, Staudtstra\ss e 7, 91058 Erlangen, Germany}

\begin{abstract}
This study looks at the finite-dimensional adiabatic evolution influenced by weak perturbations, extending the analysis to the asymptotic time limit. Beginning with the fundamentals of adiabatic transformations and time-dependent effective Hamiltonians, we intuitively derive the Bloch equation. Our investigation of the solutions of the Bloch equation underscores the critical role of initial conditions and the assured existence of solutions, revealing the intricate link between leakage phenomena and the Bloch transformation. Numerical and analytical evaluations demonstrate that the leakage can remain small \textit{eternally}. That is, a system that starts in a particular eigenspace of the strong generator remains in the same respective eigenspace for arbitrary long times with an error of $\mathcal{O}(\gamma^{-1})$, where $\gamma$ describes the ratio between the strength of the system's strong Hamiltonian and the perturbation. 
\end{abstract}

\maketitle

\section{Introduction}\label{ch.intro}
The landscape of quantum mechanics and the realm of quantum evolution present both profound theoretical challenges and expansive potential for technological advancement. One of the key challenges in the field is the problem of quantum leakage, that is, unwanted transitions that occur when a quantum system, ideally evolving within its initial eigenspace, transitions out of this due to non-adiabatic effects or perturbations. This phenomenon is not just a theoretical curiosity; it has significant implications for quantum computing~\cite{aliferis_fault-tolerant_2007,motzoi_simple_2009,miao_overcoming_2023} and other quantum technologies where coherent evolution and the integrity of quantum states in a low-energy subspace are of paramount importance.

Adiabatic control emerges as a compelling strategy to mitigate quantum leakage. By ensuring that the evolution of a quantum system remains adiabatic, one can theoretically maintain the system within its initial eigenspace. This approach, however, has its complexities. The requirement for adiabaticity~\cite{kato1950} — particularly, the slow variation of system parameters compared to the inverse of the energy gap — imposes operational constraints and necessitates a deeper understanding of system dynamics. This is especially important for long-term evolution, as the current bounds, an example found in Theorem~4 in Ref.~\cite{burgarth_one_2022}, grow with the total time of the evolution.

This study specifically investigates finite-dimensional driven quantum systems operating within the weak coupling limit for comparatively large times — an essential regime for a host of practical applications where interactions are subtle yet non-negligible. In this context, bounding the leakage out of the eigenspaces of the system's Hamiltonian (drift) for an extended period of driving becomes a crucial objective. Establishing such bounds contributes to the foundational understanding of quantum systems under adiabatic evolution and informs the design and optimization of quantum operations, particularly in scenarios where quantum systems must maintain coherence over long timescales.

Our investigation serves as an extension of the work of Burgarth \textit{et al.}\ from Ref.~\cite{eternal}, and Szab\'o \textit{et al.}\ from Ref.~\cite{szabo_robust_2025}, which prove that adiabaticity holds \textit{eternally} in autonomous systems. They show that for a static strong drift and weak drive Hamiltonians with a $\gamma$ characteristic relative strength, the system remains in each eigenspace of the strong generator for \textit{arbitrary long driving} with a leakage of only $\mathcal{O}(\gamma^{-1})$. Relatedly, Facchi \textit{et al.} in Ref.~\cite{facchi_robustness_2025} analyze the long-time stability of time-independent quantum systems with relatively bounded perturbations. This paper aims to extend the analysis to time-dependent Hamiltonian systems.

To address the challenge of analyzing the leakage for asymptotic times, our approach concentrates on the theory of effective generators. Effective Hamiltonians have been comprehensively studied in quantum mechanics and condensed matter physics as a key tool to understand complex many-body systems. The distinctive advantage of these Hamiltonians lies in their ability to achieve a separation of timescales. This feature allows isolating and concentrating on a system's relevant degrees of freedom. The separation is particularly important in understanding the dynamics and interactions in quantum systems, where direct analytical solutions often need to be more attainable. Furthermore, effective Hamiltonians provide a framework for perturbative solutions, simplifying otherwise intractable quantum many-body problems. Through this perturbative approach, they unveil an approximate but insightful and hopefully faithful picture of the system's behavior, making them indispensable in theoretical analysis and practical applications. Illustrative examples of effective Hamiltonians underscore their diverse applications and methodologies. The Zeno Hamiltonian~\cite{facchi_quantum_2008,facchi_quantum_2010}, for instance, operates under the premise of an infinitely strong separation of timescales, describing Hamiltonian evolution projected on its individual eigenspaces. Both static effective Hamiltonians~\cite{venkatraman_static_2022} and Floquet Hamiltonians~\cite{floquet_sur_1883, Schmidt_Floquet_2019} adeptly describe crucial dynamics in systems subjected to fast oscillating drives. The Schrieffer-Wolff Hamiltonian~\cite{schrieffer_relation_1966,bravyi_schriefferwolff_2011,goldin_nonlinear_2000,theis_counteracting_2018,malekakhlagh_first-principles_2020,malekakhlagh_time-dependent_2022} presents a perturbative solution that maintains unitarity, thereby preserving the physicality of the approximative solution. These examples highlight effective Hamiltonians' versatility and indispensability in deciphering quantum systems' behaviors across various scenarios. The selection of the most appropriate effective Hamiltonian is highly dependent on the specific problem at hand and on the study's goals; however, the debate of which effective Hamiltonian is optimal for each problem is outside the scope of this paper.

Additionally, the Bloch method initially proposed in Ref.~\cite{bloch_sur_1958} focuses on the similarity transformation itself rather than the effective Hamiltonian. The method - renowned for its application in solid-state physics and spin systems under the influence of external fields and internal interactions - employs wave operators that map the eigenfunctions of the strong generator to the evolution of the complete system. As a result, the effective Hamiltonian mirrors the spectral characteristics of the original Hamiltonian, maintaining the eigenvalues while achieving a diagonal evolution. Iterative expansions for the time-independent Bloch wave operator were proposed by Lindgren~\cite{lindgren_rayleigh-schrodinger_1974}, Durand~\cite{durand_direct_1983}, and Jolicard~\cite{jolicard_effective_1987}, while Jolicard and Austin first presented a recursive solution to the time-dependent problem in Ref.~\cite{jolicard_recursive_1991} and later Jolicard et al. in Ref.~\cite{jolicard_new_1999}. The two-part topical review by Jolicard and Killingbeck~\cite{jolicard_bloch_2003,killingbeck_bloch_2003} on the time-independent and time-dependent Bloch equation presents an exhaustive theoretical discussion on the topic with several examples of the method. However, this literature falls short of discussing the existence of the solution to the Bloch equation, discussing the behavior of the solution for different initial conditions, or examining the error bounds and thus the asymptotic leakage.

Our study focuses first on defining the Bloch wave operator and rigorously deriving the time-dependent Bloch equation while accounting for the subtleties of adiabatic processes in the presence of perturbation and the associated leakage phenomena in a time-dependent context. The derivations presented here incorporate these critical aspects, extending the applicability of the Bloch equation and enhancing its descriptive power in quantum systems experiencing adiabatic evolution, including the asymptotic time limit. We use this refined tool to link the leakage and the norm of the Bloch wave vector to provide more accurate predictions of quantum transitions and to pave the way for developing improved adiabatic control methods that effectively bound quantum leakage for long-duration quantum operations. For the sake of simplicity, we will restrict ourselves to finite-dimensional closed quantum systems, noting that substantial progress on long-term robustness against time-independent unbounded perturbations was recently made.

\section{Deriving the Time-Dependent Bloch Equation} \label{ch:2}
This section is dedicated to deriving the time-dependent Bloch equation for the unitary dynamics $F(t,t_0)$ from time $t_0$ to $t$, solving the Schr\"odinger equation
\begin{equation} 
    \label{eq:F_diff}
    \dot{F}(t,t_0)=[\gamma \bar{B}(t)+\bar{C}(t)]F(t,t_0), 
\end{equation}
 where $\gamma$ embodies the adiabatic parameter signifying the rate of change with respect to the system's intrinsic timescales, $\bar{B}(t)$ stands for the system's (strong) Hamiltonian, and $\bar{C}(t)$ symbolizes the (weak) drive, both generators being skew-Hermitian. In this context, the drive is treated as a perturbation, acknowledging its secondary role compared to the primary system dynamics. This equation is instrumental, serving as the backbone of our investigation, as it encapsulates the system's evolution over time, influenced by its internal dynamics (drift) and external driving forces.

 The analysis is grounded in standard assumptions compatible with the adiabatic theorem~\cite{kato1950,nenciu_linear_1993,burgarth_generalized_2019,gaeta_perturbation_2022,burgarth_one_2022}. 
 The skew-Hermitian time-dependent strong generator has the spectrum $\{b_k(t)\}$ and corresponding orthogonal eigenprojectors $\{P_k(t)\}$:
\begin{equation} 
    \label{eq:B_sum}
    \Bar{B}(t) = \sum_{k}^{} b_k(t)P_k(t). 
\end{equation}
The adiabatic assumptions are that the eigenvalues of $\Bar{B}(t)$ remain non-crossing $(b_k(t)\neq b_l(t), \ \forall k\neq l)$, continuously differentiable functions of time, and the projections $P_k(t)$ are continuously differentiable twice.

In the time-dependent scenario, the block representation, and consequently the eigenspaces of the generator $\Bar{B}(t)$, are time-dependent. Therefore, it becomes necessary to dynamically transform the evolution operator and employ the adiabatic theorem to describe quantum evolution in a constant diagonal representation over time. 

\subsection{The Adiabatic Frame}

Kato's work on the adiabatic theorem~\cite{kato1950} provides the theoretical foundation for transforming the problem into the natural frame of the strong generator. In this picture, the drift has the right block structure that is independent of time.

An adiabatic transporter can be defined as the unitary
\begin{equation} 
    \label{eq:W}
    W(t_f,t_0) = \mathcal{T} \exp\!\left( \int_{t_0}^{t_f} dt\, A(t) \right),
\end{equation}
where
\begin{equation} 
    \label{eq:A}
    A(t) = \frac{1}{2} \sum_k[\dot{P_k}(t),P_k(t)].
\end{equation}
From the adiabatic theorem, it follows that $W(t)$ satisfies the intertwining relations
\begin{equation} 
    \label{eq:WP}
    W(t,t_0)P_k(t_0) = P_k(t)W(t,t_0),
\end{equation}
enabling the definition of a new generator with a time-wise constant block structure:
\begin{equation} 
    \label{eq:B_tilde}
    B(t,t_0) = W\inv(t,t_0)\Bar{B}(t)W(t,t_0) = \sum_{k}^{} b_k(t)P_k(t_0).
\end{equation}

Now, one can rewrite the original evolution operator $F(t_f,t_0)$ between the initial time $t_0$ and the final time $t_f$ as:
\begin{widetext}
\begin{equation} 
    \label{eq:1}
    \begin{split}
    F(t_f,t_0)&=\mathcal{T} \exp\!\left(\int_{t_0}^{t_f} dt\, [\gamma \Bar{B}(t)+\Bar{C}(t)] \right)   \\
    &= \mathcal{T} \exp\!\left(\int_{t_0}^{t_f} dt\, [\gamma \Bar{B}(t)+\Bar{C}(t)-A(t)+A(t)]\right) \\
    &= W(t_f,t_0) \mathcal{T} \exp\!\left( \int_{t_0}^{t_f} dt\,W\inv(t,t_0)[\gamma \Bar{B}(t)+\Bar{C}(t)-A(t)]W(t,t_0) \right)\\
    &= W(t_f,t_0) \mathcal{T} \exp\!\left( \int_{t_0}^{t_f} dt\,[\gamma B(t,t_0)+C(t,t_0)] \right) \\
    &=W(t_f,t_0)M(t_f,t_0)
    \end{split}
\end{equation}
\end{widetext}
with the modified weak drive $C(t,t_0)$ having components originating from the transformation of $\Bar{C}(t)$ and also from the derivative of the transporter $W(t,t_0)$:
\begin{equation} 
    \label{eq:C_tilde}
    C(t,t_0) = W\inv(t,t_0)[\Bar{C}(t)-A(t)]W(t,t_0).
\end{equation}
It should be noted that the skew-Hermitian nature of the weak generator in this adiabatic frame is a direct consequence of the Hermiticity of the projectors $\{P_k(t)\}$. It is, therefore, convenient to define a skew-Hermitian generator $H(t,t_0) = \gamma B(t,t_0)+C(t,t_0)$ and the respective unitary $M(t_f,t_0)$ that will be referred to as the Hamiltonian and the full evolution operator in the rest of this section that explores the general theory of effective Hamiltonians, their corresponding transformation operators and their relation to leakage in this adiabatic frame with time-wise constant block structure determined by the projectors $\{P_k(t_0)\}$ at the initial time $t_0$.

Without loss of consistency, for clarity, the starting time parameter $t_0$ in the following sections will be suppressed for all operators, as each has this as its last argument. For example, $P_k$ will imply $P_k(t_0)$, and $H(t)$ will denote $H(t,t_0)$, however, eigenvalues like $b_k(t)$ retain their original notation.

\subsection{Dynamical Transformations}
As mentioned, this section focuses on the full-evolution operator:
\begin{equation} 
    \label{eq:M}
    M(t_f)=\mathcal{T} \exp\!\left(\int_{t_0}^{t_f} dt\,H(t)   \right),
\end{equation}
that undergoes a transformation with a yet very general dynamical transformation operator $U(t)$ that is required to be invertible and differentiable on the whole of its domain. After transformation, the effective evolution $M_{\text{eff}}(t_f)$ and its generator $H_{\text{eff}}(t)$ are defined in the canonical way:
\begin{gather} 
       M(t_f) = U(t_f)M_{\text{eff}}(t_f)U^{-1}(t_0), \label{eq:Meff} \\
       H_{\text{eff}}(t) = U^{-1}(t)H(t)U(t)-U^{-1}(t)\dot{U}(t).  \label{eq:Heff}
\end{gather}

Ideally, an effective generator should yield a faithful evolution, meaning that $M_{\text{eff}}(t)$ closely approximates $M(t)$. Specifically, in addressing the weak coupling problem outlined in Eq.~\eqref{eq:F_diff}, this approximation could be formulated as either a power series in terms of the adiabatic parameter $\gamma^{-1}$ or potentially as an exponential function of this to describe an exponentially accurate effective approximation.

The effective generator $H_{\text{eff}}(t)$ and an initial condition $U(t_0)$ on this time-dependent transformation do therefore determine the full dynamical transformation and vice versa. Rearranging Eq.~\eqref{eq:Heff}, it is easy to see that the dynamical transformation follows the differential equation 
\begin{equation} 
    \label{eq:U-dot}
    \dot{U}(t)=H(t)U(t)-U(t)H_{\text{eff}}(t).
\end{equation}

Many descriptions of quantum evolution by effective Hamiltonians fit into the picture presented above; however, some are especially interesting in the context of the problem of adiabatic controls and weak coupling presented in Eq.~\eqref{eq:F_diff}.

\subsection{Leakage and the Block-Diagonal Effective Hamiltonians}
In the context of transitions out of the eigenspaces of the strong generator, it is advantageous to look at block-diagonal effective generators and evolutions. Mathematically the leakage out of a given space defined by its projection $P_k(t)$ is the norm of $\|[\mathbbm{1}-P_k(t)]F(t)P_k(t_0)\|$ which for all unitarity invariant norms is equivalent to $\|[\mathbbm{1}-P_k(t_0)]M(t)P_k(t_0)\|$. This can be upper bounded by the distance of the true evolution $M(t)$ to a block-diagonal effective evolution $M_{\text{eff}}(t)$ that expresses no leakage.

Following our discussion of the dynamical transformation, one should also link the leakage to this operator. That is, if the transformation is close to identity, then obviously the true evolution is quasi-diagonal, and the leakage is small. The mathematically rigorous statement is that
\begin{equation} 
    \label{eq:leakage}
\|M(t)-M_{\text{eff}}(t)\| \leq \frac{2\delta}{1-\delta},
\end{equation}
where $\delta$ is a uniform bound of the distance of the transformation operator to identity: $\|U(t)-\mathbbm{1}\| \leq \delta$, $\forall t \in [t_0,t_f]$. See Appendix~\ref{sec:appendix} for the derivation of this bound.

An effective Hamiltonian and evolution describe no leakage for all times $t$ from the initial time $t_0$ to the final time $t_f$ if they are endowed with the same block structure as the strong generator $B(t)$, i.e., for all subspaces $k$ 
\begin{equation}
    \label{HM_block}
   [H_{\text{eff}}(t),P_k] = [M_{\text{eff}}(t),P_k]=0.
\end{equation}

For this system, the operators 
\begin{gather} 
       M_{\text{eff,}k}(t) \equiv  P_kM_{\text{eff}}(t)P_k, \quad M_{\text{eff}}(t) = \sum_k M_{\text{eff,}k}(t),    \label{eq:Meffk} \\
      U_k(t) \equiv U(t)P_k, \quad U(t) = \sum_k U_k(t)    \label{eq:Uk} 
\end{gather}
act on the specific time-independent subspaces of $P_k$. Rearranging Eq.~\eqref{eq:Meff}, it is easy to express the transformation operator and the wave vector in terms of the effective and true evolutions:
\begin{gather} 
      U(t) = M(t) U(t_0) M_{\text{eff}}^{-1}(t) ,    \label{eq:Ufsol} \\
      U_k(t) = M(t) U(t_0) M_{\text{eff,}k}^{-1}(t) ,    \label{eq:Ukfsol} 
\end{gather}
where the second equation follows from multiplying the first one on the right-hand side by the projector $P_k$. In the second equation, the inverse is the appropriate pseudo inverse acting on the subspace such that $M_{\text{eff,}k}^{-1}(t)M_{\text{eff,}k}(t)=M_{\text{eff,}k}(t)M_{\text{eff,}k}^{-1}(t)=P_k$, that acts as an inverse on the subspace of $P_k$, and it is $0$ otherwise. Remark that for all non-zero $M_{\text{eff,}k}(t)$, this inverse exists and is well determined. Moreover, the block structure of the effective evolution is inherited by its inverse as well. Further discussion on the existence of the inverse is included at the end of the next subsection.

Multiplying Eq.~\eqref{eq:U-dot} from the right-hand side with the projectors $P_k$, it is possible to rewrite the differential equation for the dynamical transformation acting on the explicit subspaces describing diagonal evolution
\begin{equation} 
    \label{eq:Uk-dot}
    \dot{U}_k(t)=H(t)U_k(t)-U_k(t)H_{\text{eff}}(t).
\end{equation}
This will prove useful when talking about more specific effective Hamiltonians, especially for deriving the Bloch equation.

\begin{figure}[] % [h] to fix position 
     \centering
     \includegraphics[width=0.5\textwidth]{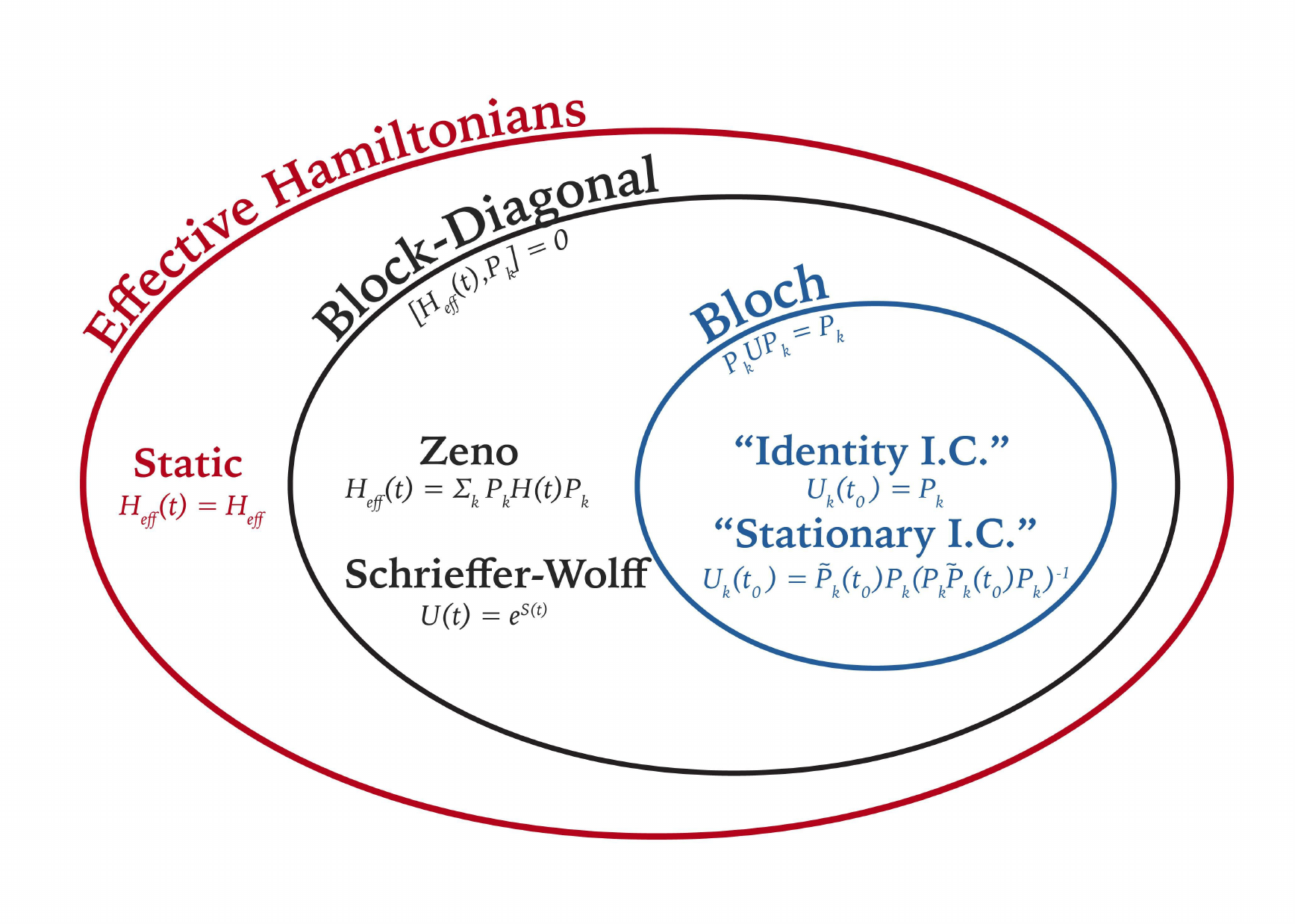}
     \caption{A Venn diagram of effective time-dependent Hamiltonians: Diagonal, Zeno, and Schrieffer-Wolff and Bloch, highlighting two specific solutions of the latter, based on the initial conditions.}
    \label{fig:effham}
\end{figure}

The effective diagonal Hamiltonians studied in the literature are obtained by restrictions on the structure of the Hamiltonian or on the dynamical transformation itself. Figure \ref{fig:effham} highlights the relationship between the most studied ones. The Zeno Hamiltonian~\cite{facchi_quantum_2008,facchi_quantum_2010} describes a first-order approximation to the weak coupling model and is elegantly achieved by projecting the effective generator onto each subspace $H_{\text{eff}}(t)=\sum_k P_k H(t) P_k$. The perturbative Schrieffer-Wolff generator~\cite{bloch_sur_1958,lindgren_rayleigh-schrodinger_1974,durand_direct_1983,jolicard_effective_1987,jolicard_recursive_1991,jolicard_new_1999,jolicard_bloch_2003,killingbeck_bloch_2003,eternal} results in postulating the form of the wave operator to be the exponential $U(t)=e^{S(t)}$ for a purely off-diagonal, time-dependent, skew-Hermitian generator, i.e.~$P_kS(t)P_k=0$ for all $k$. Hence, in this form, both the transformation and the effective evolution become unitary, which can have a potential advantage in preserving physical properties. The following subsection studies the particularities of the family of time-dependent Bloch generators, highlighting the insights that this has on understanding quantum leakage.

\subsection{The Bloch Generators}
In the extant body of literature, the derivation of the time-dependent Bloch equation is usually approached from the wave-operator perspective, which, from the perspective of scattering processes, transforms an incoming wave function into an outgoing one. However, one can arrive at it easily from the perspective of the theory of effective diagonal Hamiltonians presented above. We define the Bloch condition for the Bloch wave operator simply as
\begin{equation} 
    \label{eq:BC}
    P_kU(t)P_k=P_k,
\end{equation}
for all times $t$ on its domain. The transformation does not include any dynamics inside the subspace when acting on one, $P_k$, but only has terms related to transitions. In the diagonal representation of the strong generator, the wave operator acts trivially inside the blocks and potentially has other non-zero terms outside the blocks.

Multiplying the differential equation~\eqref{eq:Uk-dot} from the left-hand side with $U_k(t)$, and recognizing that $P_k\dot{U}_k(t)P_k=0$ as inferred from the Bloch condition from above follows that $0=U_k(t)H(t)U_k(t)-U_k(t)H_{\text{Bloch}}(t)$. Subsequently, substituting this identity to the differential equation of the wave vector results immediately in the well-established time-dependent Bloch equation: 
\begin{equation} 
    \label{eq:BE}
    \dot{U}_k(t)=H(t)U_k(t)-U_k(t)H(t)U_k(t),
\end{equation}
which is a non-linear operator Riccati equation~\cite{abou-kandil_matrix_2003}. Significantly, its solution solely depends on the full generator $H(t)$ and the initial condition $U(t_0)$, while the effective Hamiltonian can be calculated thereafter. The solutions to this equation have been studied in a multitude of mathematical papers and textbooks for a range of examples. Closed-form solutions are only known for a handful of examples; however, both numerical techniques~\cite{abou-kandil_matrix_2003}, as well as recursive methods~\cite{jolicard_recursive_1991} have been developed.

The solution of the Bloch equation can be expressed in terms of the full evolution
\begin{equation} 
    \label{eq:BS}
    U_k(t)=M(t)U_k(t_0)[P_kM(t)U_k(t_0)P_k]^{-1},
\end{equation}
equivalently
\begin{equation} 
    \label{eq:BSfull}
    U(t)=M(t)U(t_0)\sum_k[P_kM(t)U_k(t_0)P_k]^{-1},
\end{equation}
where the inverse is once again the unique inverse associated with the projector $P_k$. To see this, first substitute the diagonal solution~\eqref{eq:Ukfsol} into the Bloch condition~\eqref{eq:BC}. Explicitly showing the relevant projectors, this condition can be written as 
\begin{equation} 
    \label{eq:BChelp}
    P_k[M(t) U(t_0)]P_k [P_kM_{\text{Bloch,}k}(t)P_k]^{-1} = P_k.
\end{equation}
From the uniqueness of the inverse $M_{\text{Bloch,}k}^{-1}(t)$ in the block $P_k$ follows the exact expression for the effective diagonal Bloch evolution:
\begin{equation} 
    \label{eq:MBloch}
    M_{\text{eff,}k}(t)=M_{\text{Bloch,}k}(t)=P_k M(t) U(t_0) P_k.
\end{equation}
The Bloch solution~\eqref{eq:BS} is just substituting this Bloch evolution~\eqref{eq:MBloch} to the general diagonal solution~\eqref{eq:Ukfsol}. For example, in Ref.~\cite{jolicard_effective_1995}, this Bloch solution is the defining equation of the time-dependent wave operator, and the general theory and the Bloch equation are derived from this starting point, but without recognizing that the transformation could have an initial value that is different from identity.  

The time-dependent Bloch equation can be solved explicitly as a Riccati equation using Radon's lemma~\cite{abou-kandil_matrix_2003} that links the non-linear differential equation to an associated linear system. (see Appendix~\ref{sec:appendixRadon}). Nevertheless, after carrying out the ensuing simplifications, this route reproduces the same conclusions already obtained above and offers no additional insights.

An important aspect of discussions of the Bloch equation concerns its initial conditions. While most literature focuses on the equation's solutions, there is an implicit assumption, for example, in the recursive solution of Ref.~\cite{jolicard_new_1999} and also in the review~\cite{killingbeck_bloch_2003}, that the transformation at the initial time $t_0$ is the identity, denoted $U(t_0)=\mathbbm{1}$, or equivalently $U_k(t_0)=P_k$. We are going to refer to this as the ``identity'' initial condition. Alternatively, it is also possible to define the initial condition to be the solution to the time-independent Bloch equation at the relevant starting time. Under this postulate, the derivative of the wave operator is vanishing at this initial time. Hence, it is convenient to call it the ``stationary'' initial condition. The solution to the time-independent problem has been presented many times; for example, see Durand's~\cite{durand_direct_1983} approach. In operator form, this initial condition reads $U_k(t_0)= \Tilde{P}_k(t_0)P_k[P_k\Tilde{P}_k(t_0)P_k]^{-1}$, where $\Tilde{P}_k(t)$ represents the projector onto the direct sum of the eigenspaces of the full Hamiltonian $H(t)=\gamma B(t)+C(t)$ belonging to the eigenvalue $b_k(t)$ of the strong generator. These initial conditions can be likened to the quantum Bloch equivalent of the well-known Dirichlet and Neumann boundary conditions. The first directly gives the problem's boundary value, while the second sets the value of the derivative to zero. There are infinitely many other boundary values that could be considered, all resulting in different effective Hamiltonians, finding the optimal of which is beyond the scope of this analysis.

It is also important to discuss the existence of the Bloch solution of the transformation $U(t)$ --- when considering an infinitely prolonged period --- and the blow-up conditions, in which case it would grow to be unbounded. Arguably, there is a logical connection between leakage, the existence of the inverse of the diagonal evolution, and the existence of the transformation operator. The previous subsection started with postulating the existence of $U^{-1}(t)$, from which, through the definition of the effective evolution, follows that both $M_{\text{Bloch}}(t)$ and $M_{\text{Bloch}}^{-1}(t)$ are well defined for all times. Conversely, from Eq.~\eqref{eq:Ufsol}, one can see that the existence of an effective evolution and its inverse imply the existence and invertibility of a wave operator.

Additionally, in scenarios where the leakage is minimal, on the order of $\mathcal{O}(\gamma^{-1}))$, the solution~\eqref{eq:BS} suggests that any suitable initial condition $\|U(t_0) - \mathbbm{1}\| = \mathcal{O}(\gamma^{-1}))$ will yield a Bloch wave operator that is similarly proximate to the identity. Finding such an initial condition is straightforward. Both the 'identity' and the `stationary' boundary conditions satisfy this requirement; the latter was demonstrated, for instance, in the paper~\cite{eternal}. Conversely, as per the definition of leakage in Eq.~\eqref{eq:leakage}, the inverse of this assertion is self-evident. If a minor transformation exists, such that$\|U(t) - \mathbbm{1}\| = \mathcal{O}(\gamma^{-1})$, then a Bloch evolution approximating the true evolution within $\mathcal{O}(\gamma^{-1})$ must also exist, implying that the leakage is correspondingly small. Figure \ref{fig:trinity} visually encapsulates this triad of equivalent concepts.

\begin{figure}[] % [h] to fix position 
     \centering
     \includegraphics[width=0.3\textwidth]{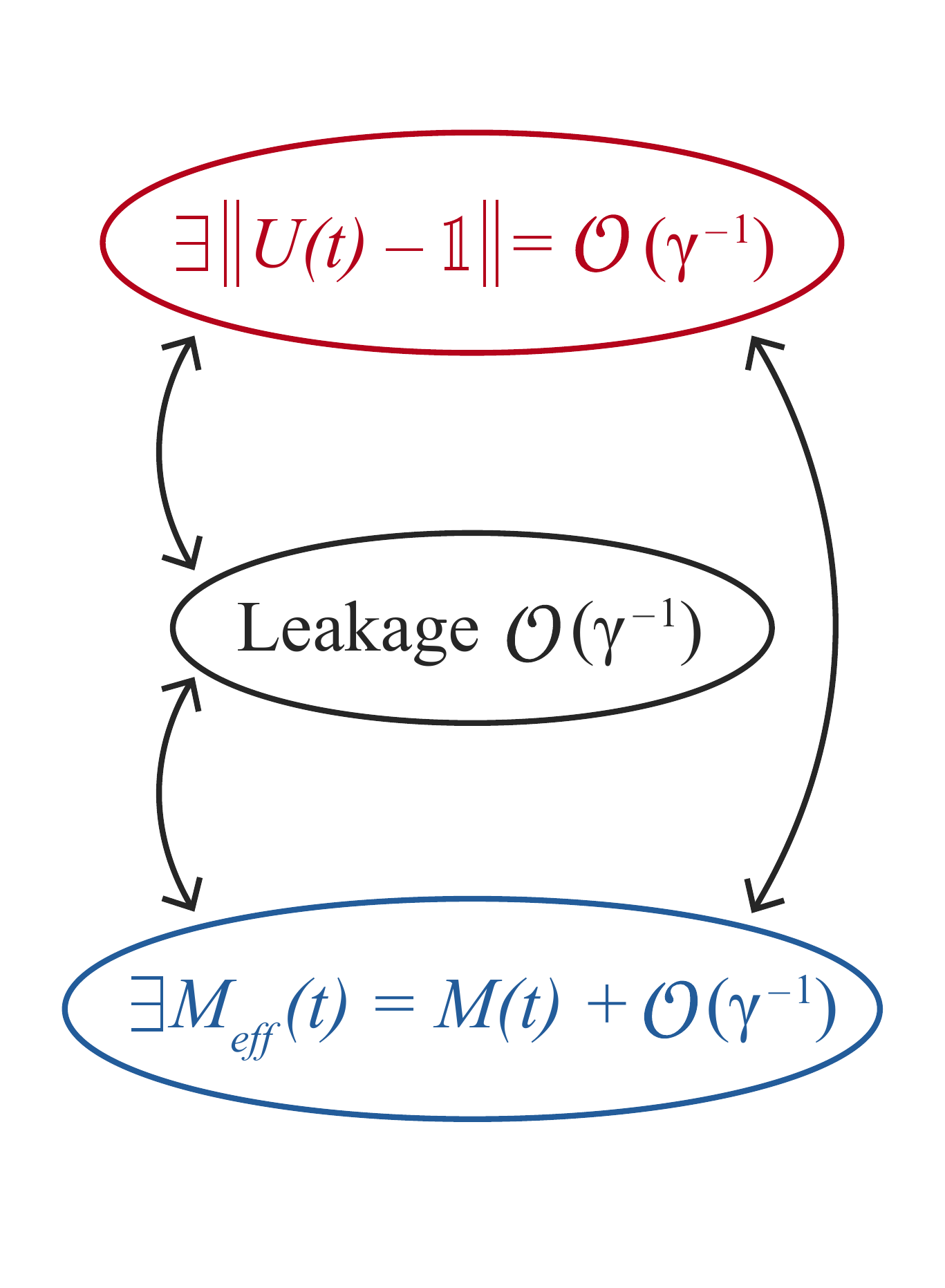}
     \caption{The triad of the leakage - Bloch Hamiltonian $M_{\text{eff}}(t)$ - wave operator $U(t)$, highlighting the three equivalent statements. }
     \label{fig:trinity}
\end{figure}

\subsection{Constructing a Unitary Effective Evolution}
A drawback of the Bloch transformation is that it is not a unitary change of basis; hence, the resulting evolution operator $M_{\text{Bloch}}(t)$ is, in general, not unitary either. Given a uniformly bounded Bloch operator, $\|U(t)-\mathbbm{1}\| \leq \delta \ll1$, with an appropriate initial condition such that $[U^\dag(t_0)U(t_0),P_k] = 0 $, it is, however, possible to construct a unitary transformation that also block-diagonalizes the evolution $M(t)$. Inspired by the polar decomposition of the Bloch operator, define 
\begin{equation}
    V(t)=U(t)[U^\dag (t)U(t)]^{-1/2},
\end{equation}
that is unitary by construction. 

The correct choice of initial conditions, such that $U^\dag(t_0)U(t_0)$ is block-diagonal ensures that
\begin{align}\label{eq:UdagU}
    U^\dag(t)U(t)=\sum_{k}& [P_kU_k^\dag(t_0)M^\dag(t)P_k]^{-1}  \\
    &{}\times U^\dag(t_0)U(t_0) 
    \sum_{l}[P_lM(t)U_k(t_0)P_l]^{-1}
\end{align}
remains block-diagonal for all times. Such an initial condition is not difficult to find; both the ``identity'' and ``stationary'' initial conditions satisfy this criterion. The bound on the Bloch transformation guarantees that this unitary is always well defined and also that it is close to identity:
\begin{equation} \label{eq:Vminus1}
    \|V(t)-\mathbbm{1}\| \leq \frac{1+\delta}{\sqrt{1-2\delta-\delta^2}}-1.
\end{equation}
The proof is provided in Appendix~\ref{sec:appendix3}.

This construction makes it immediately clear that, in the given frame, the effective evolution is block-diagonal:
\begin{widetext}
\begin{align} 
    \label{eq:MSW}
    V^{-1}(t)M(t)V(t_0) &= [U^\dag (t)U(t) ]^{-1/2} U^\dag(t)M(t)U(t_0) [U^\dag (t_0)U(t_0)]^{-1/2} \nonumber \\
    &= [U^\dag (t)U(t) ]^{-1/2} U^\dag(t)U(t)M_{\text{Bloch}}(t) [U^\dag (t_0)U(t_0) ]^{-1/2}  \nonumber \\
    &= [U^\dag (t)U(t) ]^{1/2} M_{\text{Bloch}}(t) [U^\dag (t_0)U(t_0) ]^{-1/2}.  \\
\end{align}
\end{widetext}
In the second identity we used the definition of the Bloch evolution~\eqref{eq:Meff} to get to a product of three diagonal operators in the last line. 

Subsequent sections detail this framework of the Bloch wave operator, its initial conditions, and how asymptotic leakage manifests in finite-dimensional problems that are of significant interest to the quantum community.

\section{Examples} \label{ch:3}

This section explores a couple of noteworthy time-dependent problems to demonstrate the efficacy of our approach in quantifying leakage as the timescale approaches infinity. The examination begins with a detailed look at one of the exactly solvable such problems, discussing insights into its dynamics. Furthermore, the second part of the chapter looks at a case significant for transmon qubits controlled with long pulses. While not exactly solvable, this example is studied using numerical methods. The focus is on investigating the Bloch wave operator and the behavior of the bound on the leakage under various initial conditions. This dual approach enhances the understanding of these specific cases and highlights the broader implications and applications within driven quantum systems.

\subsection{The Landau-Zener Model}

We solve the Landau-Zener model~\cite{zeener_non-adiabatic_1932, ho_simple_2014, matus_analytic_2023} analytically for the true evolution $M(t)$ and use the Bloch solution~\eqref{eq:BS} to show that bounding the Bloch transformation reproduces the right bound of the non-adiabatic transition rate known in literature.   

The canonical model from the perspective of the notation of the study focusing on strong and weak generators reads: 
\begin{gather} 
       \Bar{B}(t) = -i(X+tZ), \label{eq:LZB} \\
       \Bar{C}(t) = 0,  \label{eq:LZC}
\end{gather}
where $X$, $Z$, and later $Y$ will also be used, are the Pauli operators. 

The eigenvalues of the strong generator $\Bar{B}(t)$ 
\begin{equation} \label{eq:LZ1-bs}
b_{\pm}(t) = \pm i \sqrt{1+t^2},
\end{equation}
do not cross, and the eigenprojections are continuously differentiable twice
\begin{equation} \label{eq:LZ1-Ps}
P_\pm(t) = \frac{1}{2}\left(\mathbbm{1} \mp \frac{X+tZ}{\sqrt{1+t^2}}\right),
\end{equation}
\begin{equation} \label{eq:LZ1-Pdots}
\dot{P}_\pm(t) = \mp \frac{1}{2} \frac{Z-tX}{\sqrt{1+t^2}},
\end{equation}
such that
\begin{equation} \label{eq:LZ1-PdotPcom}
[\dot{P}_\pm(t),P_\pm(t)]=\frac{1}{2}\frac{i Y}{1+t^2},
\end{equation}
therefore, the adiabatic theorem is applicable, and the adiabatic transporter is found by Eqs.~\eqref{eq:W} and~\eqref{eq:A} to be
\begin{align}
\label{eq:LZ1-W}
W(t,t_0) ={}&\mathcal{T} \exp\!\left(\int_{t_0}^{t} ds\, A(s)  \right)  \\ = {}&\exp\!\left(\frac{i Y}{2}[\arctan(t)-\arctan(t_0)]\right)  \\
={}& \cos\!\left(\frac{\arctan(t)-\arctan(t_0)}{2}\right)\mathbbm{1}\\&{}+i\sin\!\left(\frac{\arctan(t)-\arctan(t_0)}{2}\right)Y,
\end{align}
which has the asymptotic limit $\underset{t\rightarrow\infty}{\lim}W(t,-t)=iY$.

The new strong generator, with time-independent eigenprojections, becomes
\begin{align}
\label{eq:LZ1-B}
B(t,t_0)
&=\sum_\pm b_{\pm}(t) P_\pm(t_0)  \\ 
&=\sum_\pm \pm \frac{i \sqrt{1+t^2}}{2}\left(\mathbbm{1} \mp \frac{X+t_0Z}{\sqrt{1+t_0^2}}\right),
\end{align}
while the weak generator, using Eq.~\eqref{eq:C_tilde} is 
\begin{equation} \label{eq:LZ1-C}
C(t,t_0) = - \frac{iY}{2(1+t^2)}.
\end{equation}

First, we solve the exact evolution to study the properties of the wave operator $U_k(t,t_0)$.

\subsubsection{Solving Landau-Zener}

The true evolution, $F(t,t_0)$, is solved by setting up the Landau-Zener problem as 
\begin{equation} \label{eq:LZ1-de1}
\begin{pmatrix}
\dot{\phi}(t) \\
\dot{\psi}(t)
\end{pmatrix}=-i
\begin{pmatrix}
\gamma t & \gamma \\
\gamma & - \gamma t
\end{pmatrix}
\begin{pmatrix}
\phi(t) \\
\psi(t)\end{pmatrix},
\end{equation}
from which a second order differential equation for $\phi(t)$ can be derived:
\begin{equation} \label{eq:LZ1-de2}
\begin{split}\ddot{\phi}(t) & =-i\gamma\phi(t)-i\gamma t\dot{\phi}(t)-i\gamma\dot{\psi}(t)\\
 & =-i\gamma\phi(t)-\gamma^{2}t^{2}\phi(t)-\gamma^{2}\phi(t).
\end{split}
\end{equation}
By rescaling the time and coupling variables, this is equivalent to the differential equation in the standard form: 
\begin{equation}\label{eq:LZ1-LZstandard2}
\begin{split}2i\gamma\frac{d^{2}}{dz^{2}}\phi(z) & =-i\gamma\phi(z)-\gamma^{2}\frac{z^{2}}{2i\gamma}\phi(z)-\gamma^{2}\phi(z)\\
\frac{d^{2}}{dz^{2}}\phi(z) & =-\frac{1}{2}\phi(z)+\frac{z^{2}}{4}\phi(z)+\frac{i\gamma}{2}\phi(z)\\
\end{split}    
\end{equation}
\begin{equation}\label{eq:LZ1-LZstandard}
\frac{d^2}{dz^2}\phi(z)+\left(n+\frac{1}{2}-\frac{z^2}{4}\right)\phi(z) = 0,    
\end{equation}
two independent solutions of which are the parabolic cylinder functions $D_n(\pm z)$, with   
\begin{equation} \label{eq:LZ1-z}
    z = \sqrt{2\gamma} e^{+i\frac{\pi}{4}} t ,
\end{equation}
\begin{equation} \label{enq:LZ1-n}
    n = -\frac{i\gamma}{2}, \quad n \in i\mathbbm{R}^+.
\end{equation}

$\psi(s)$ is determined from Eq.~\eqref{eq:LZ1-de1}, $\psi(t) =  \frac{1}{\gamma} [i \dot{\phi}(t) - \gamma t \phi(t)] $. Using the recurrence relation of the derivative of parabolic cylinder functions~\cite{dlmf} $\frac{d}{dz}D_n(z)=-\frac{1}{2}zD_n(z)+nD_{n-1}(z)$, two independent solutions for the differential equation~\eqref{eq:LZ1-de1} are 
\begin{equation} \label{eq:LZ1-vecsol}
\begin{pmatrix}
D_{n}( \pm z) \\
\pm i\sqrt{n} D_{n-1}(\pm z)
\end{pmatrix}.
\end{equation}
For convenience, work with the normalized solution of 
\begin{align}
\label{eq:LZ1-normsol}
&\ket{{\uparrow}(t)}=
\begin{pmatrix}
\phi_{\uparrow}(t) \\
\psi_{\uparrow}(t)
\end{pmatrix} \\
&=e^{-\pi\gamma/8}
\begin{pmatrix}
D_{-i\gamma /2}(-\sqrt{2\gamma}\exp^{i\pi/4}t) \\
-\sqrt{\gamma/2}e^{i\pi/4}D_{-i\gamma /2-1}(-\sqrt{2\gamma}\exp^{i\pi/4}t)
\end{pmatrix},
\end{align}
which has an asymptotic behavior~\cite{dlmf} of 
\begin{multline}
\label{eq:LZ1-assymsol}
    \ket{{\uparrow}(t)}\rightarrow \ket{{\uparrow_{-}}(t)}=
    \begin{pmatrix}
    (-\sqrt{2\gamma}t)^{-i\gamma/2}e^{-i\gamma t^2/2} \\
    0
    \end{pmatrix}, \\
\end{multline} 

as $t\rightarrow-\infty$, and 

\begin{multline}
\ket{{\uparrow}(t)}\rightarrow \ket{{\uparrow_{+}}(t)}\\ 
= e^{-\pi\gamma/4}
\begin{pmatrix}
e^{-\pi\gamma/4}(\sqrt{2\gamma}t)^{-i\gamma/2}e^{-i\gamma t^2/2} \\
-\sqrt{\frac{2}{\gamma}}\frac{\sqrt{2\pi}}{\Gamma(i\gamma/2)}e^{-i\pi/4}(\sqrt{2\gamma}t)^{i\gamma/2}e^{i\gamma t^2/2} 
\end{pmatrix}, 
\end{multline}
as $t\rightarrow\infty$. Thus, the solution in the computational basis for $\dot{F}(t,t_0)=[\gamma \Bar{B}(t) +\Bar{C}(t)]F(t,t_0)$  with the boundary condition $F(t_0,t_0)=\mathbbm{1}$ is

\begin{equation} \label{eq:LZ1-F1F2}
F(t,t_0)=F_1(t)F_1^\dag(t_0)
\end{equation}
with
\begin{equation} \label{eq:LZ1-F1def}
F_1(t)=
\begin{pmatrix}
\phi_{{\uparrow}}(t) & -\psi_{\uparrow}^*(t)
\\
\psi_{\uparrow}(t) & \phi_{\uparrow}^*(t)
\end{pmatrix}
\end{equation}
as an auxiliary operator that is the solution for the Landau-Zener model from time $0$ to $t$.

\subsubsection{$U(t,t_0)$ for Landau-Zener}
For the exactly solvable model of Landau-Zener, it is possible to obtain a closed-form solution of the Bloch wave operator following Eq.~\eqref{eq:BS}. The full evolution $M(t_f,t_0)$ in the adiabatic frame for this problem is simply $M(t_f,t_0)=W(t_f,t_0)F_1(t_f)F_1^\dag(t_0)$, which in the computational basis is
\begin{widetext}
\begin{equation} \label{eq:MLZAS}
\lim_{t \rightarrow \infty}M(t,-t)=\left(\begin{array}{cc}
\sqrt{\frac{2}{\gamma}}\frac{\sqrt{2\pi}}{\Gamma(i\gamma/2)}e^{-\frac{i\pi}{4}}e^{-\frac{\pi\gamma}{4}}(\sqrt{2\gamma}t)^{i\gamma}e^{i\gamma t^{2}} & -e^{-\frac{\pi\gamma}{2}}\\
e^{-\frac{\pi\gamma}{2}} & \sqrt{\frac{2}{\gamma}}\frac{\sqrt{2\pi}}{\Gamma(-i\gamma/2)}e^{\frac{i\pi}{4}}e^{-\frac{\pi\gamma}{4}}(\sqrt{2\gamma}t)^{-i\gamma}e^{-i\gamma t^{2}}
\end{array}\right).
\end{equation}
\end{widetext}
Note that this is of the form
\begin{equation} \label{eq:MLZAS2}
\lim_{t \rightarrow \infty}M(t,-t)=\left(\begin{array}{cc}
\cos\phi \  e^{i\alpha(t)} & -\sin\phi\\
\sin\phi & \cos\phi \ e^{-i\alpha(t)}
\end{array}\right),
\end{equation}
with $\sin\phi=e^{-\frac{\pi\gamma}{2}}$ and $\alpha(t)$ is just a fast oscillating phase. 

In the basis of the Pauli $Z$ operator, as time extends to infinity, the projectors of the strong generator $\underset{t_0\rightarrow - \infty}{\lim}P_\pm(t_0)$ are diagonal. A notable correspondence is observed in this scenario: the projector associated with the full Hamiltonian becomes the same as that of the strong generator. Consequently, within the context of this problem, the 'identity' and `stationary' initial conditions, previously mentioned, are found to be equivalently the identity operator.
Using the solution~\eqref{eq:BS} with these initial conditions, follows that in the asymptotic limit the Bloch operator, 
\begin{equation} \label{eq:ULZAS}
\lim_{t \rightarrow \infty}U(t,-t)=\left(\begin{array}{cc}
1  & -\tan\phi \  e^{i\alpha(t)} \\
\tan\phi \  e^{-i\alpha(t)} & 1
\end{array}\right),
\end{equation}
is close to identity. The spectral norm $\underset{t\rightarrow\infty}{\lim}||U(t,-t)-\mathbbm{1}||=\tan\phi$ is  $e^{-\frac{\pi\gamma}{2}}$ for large $\gamma$. 
Therefore, in this example, our methodology reproduces a tight bound on $||U(t,-t)-\mathbbm{1}||$, which is the Landau-Zener formula.

\subsection{The Three-Level Driven System}
Contrary to the previous example, most of the problems in quantum mechanics are not exactly solvable; hence, for most problems of interest, one cannot simply invoke the Bloch solution~\eqref{eq:BS} to bound the leakage. In these examples, numerical methods can be used to compute the Bloch wave operator.

The demonstrative example presented here is the on-resonance driven three-level system with large detuning, $\gamma$~\cite{liebermann_optimal_2016}. Consider the full Hamiltonian $\hat{H}(\gamma,t)=\hat{H}_{0}(\gamma)+\hat{H}_{1}(t)$ of the drift and drive:
\begin{equation}
\label{eq:3lev1}
    \hat{H}(\gamma,t)=\begin{pmatrix}0\\
 & \omega\\
 &  & \omega (2+\gamma)
\end{pmatrix}+a\cos(\omega t)\begin{pmatrix} 
0 & -\frac{i}{2} & 0 \\
\frac{i}{2} & 0 & -\frac{i}{\sqrt{2}}\\
0 & \frac{i}{\sqrt{2}} & 0.
\end{pmatrix}
\end{equation}
Where $\omega$ is the energy scale that can be set to 1 without loss of generality, and $a$ is the driving amplitude. This amplitude could be considered time-dependent; however, for the purpose of this example, it is kept stationary.

In the interaction picture, this model transforms to the weak-coupling problem studied in the previous sections. With the transformation of $\exp(-it\hat{H}_{0}(\gamma=0))$, one obtains the skew-Hermitian generators $H(\gamma,t) = \gamma B(t)+C(t)$ of the form:
\begin{multline}
\label{eq:3lev2}
    H(\gamma,t)=-i \gamma \begin{pmatrix}0\\
 & 0 \\
 &  & 1
\end{pmatrix}+ 
\frac{a}{2}\begin{pmatrix}0 & -\frac{1}{2} & 0\\
\frac{1}{2} & 0 & -\frac{1}{\sqrt{2}}\\
0 & \frac{1}{\sqrt{2}} & 0
\end{pmatrix} \\
{}+\frac{a}{2}\begin{pmatrix}0 & -\frac{1}{2}e^{-i2t} & 0\\
\frac{1}{2}e^{i2t} & 0 & -\frac{1}{\sqrt{2}}e^{-i2t}\\
0 & \frac{1}{\sqrt{2}}e^{i2t} & 0
\end{pmatrix}.
\end{multline}

We solve the Bloch equation~\eqref{eq:BE} numerically for this tree-level system. As the figure \ref{fig:3lev1} shows, for fixed parameters $\gamma \gg a$, the wave operator exhibits a quasi-periodic structure perpetually.

Thus, it is possible to look at the maximal leakage in the asymptotic limit $t \gg \gamma$. The interest is the maximum distance from the identity of the Bloch transformation for a set amplitude $a$ as a function of the coupling $\gamma$. Following the plot \ref{fig:3lev2}, we conclude that for appropriate initial conditions, such as the ones discussed earlier, the wave operator remains $\mathcal{O}(\gamma^{-1})$ close to identity eternally. 

\begin{figure}[] % [h] to fix position 
     \centering
     \begin{subfigure}[b]{0.4\textwidth}
         \centering
         \includegraphics[width=\textwidth]{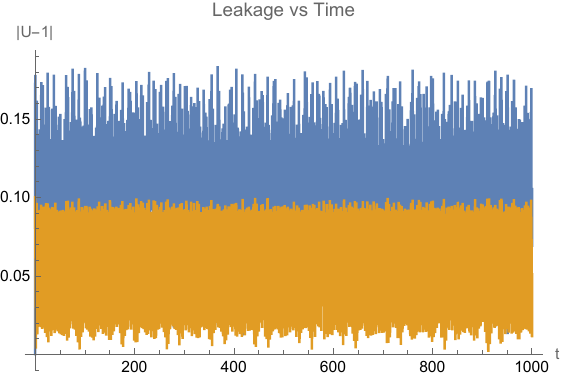}
         \caption{The Bloch operator for a set $\gamma=10.0$ and $a=1.0$, showing that the transformation is quasiperiodic. }
         \label{fig:3lev1}
     \end{subfigure}

     \centering
     \begin{subfigure}[b]{0.4\textwidth}
         \centering
         \includegraphics[width=\textwidth]{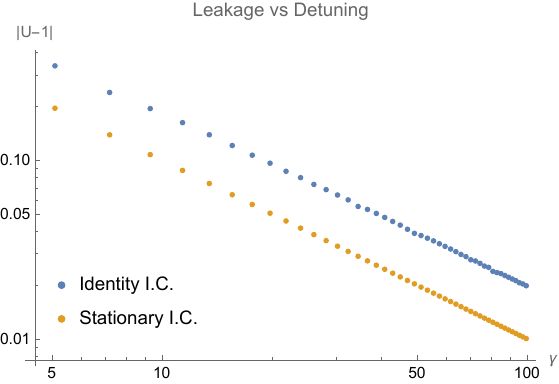}
         \caption{Log-Log plot of the wave operator for varying detuning $\gamma$ and constant drive amplitude $a=1.0$.}
         \label{fig:3lev2}
     \end{subfigure}
        \caption{The Frobenius Norm of $\|U(t) - \mathbbm{1}\|$ for the three-level System: the graphs compare the numerical solutions of the Bloch equation with the 'identity' - blue line - and 'stationary' - orange line - initial conditions. }
        \label{fig:3lev}
\end{figure}

\section{Discussions} \label{ch:4}
This study explored finite-dimensional adiabatic evolution under weak perturbations, extending our analysis to the asymptotic limit as time approaches infinity. Initially, we laid the groundwork by examining general adiabatic transformations alongside a broader perspective on time-dependent effective Hamiltonians. This foundational understanding enabled us to offer an intuitive derivation of the Bloch equation through the lens of diagonal effective Hamiltonians, improving our knowledge of the underlying physics.

The solution to the Bloch equation was investigated, focusing on the role of initial conditions and the assurance of the solution's existence. This deepened our understanding of the adiabatic process and highlighted the relationship between leakage phenomena and the Bloch transformation. Through both numerical and analytical examples, we observed that leakage remains minimal, of the order $\mathcal{O}(\gamma^{-1})$, persistently. This observation is important for quantum computing and other quantum technologies, where bounding leakage is essential to maintain system fidelity and performance. Our findings suggest a promising avenue for enhancing the robustness of the foundational theory behind adiabatic quantum computations and other applications reliant on adiabatic processes.

Building on these insights, we conjecture that in the weak coupling limit of the time-dependent problem, the leakage would consistently be $\mathcal{O}(\gamma^{-1})$. While this conjecture is grounded in our current observations, it beckons further investigation to establish a rigorous proof similar to that in Ref.~\cite{eternal} for the time-independent models. The validation of this conjecture could have profound implications for the field, potentially leading to new strategies for mitigating leakage in quantum systems.

In conclusion, our study advances the theoretical understanding of adiabatic evolution in the presence of weak perturbations and sets the stage for significant practical advancements in quantum technology. The paths laid out by this research call for further exploration, promising to unveil new horizons of quantum physics.

\begin{acknowledgements}
Z.S. was supported by the Sydney Quantum Academy.
K.Y. acknowledges support by JSPS KAKENHI Grant No.~JP24K06904 from the Japan Society for the Promotion of Science (JSPS)\@.
\end{acknowledgements}

\appendix
\begin{widetext}
\section{Bounding the Distance $\|M(t)-M_{\text{eff}}(t)\| $}\label{sec:appendix}

Here, we show the derivation of the bound~\eqref{eq:leakage} for the distance of the full evolution $M(t)$ to a block-diagonal effective one $M_{\text{eff}}(t)$. Given a real number $\delta$, such that $0\leq\delta<1$, and it uniformly bounds $ \|U(t)-\mathbbm{1}\| \leq \delta$, $\forall t$ in the domain of the operator, first note that through the triangle inequality that 
\begin{equation}\label{eq:apexU}
 \begin{split}
    \|U(t)\|&=\|U(t) - \mathbbm{1}+\mathbbm{1}\| \\
    &\leq \|U(t) - \mathbbm{1}\| +\|\mathbbm{1}\| \\
    &\leq 1+\delta.
 \end{split}
\end{equation}
Similarly, we bound $\|U^{-1}(t)\|$, introducing the operator $T(t)\equiv \mathbbm{1}-U(t)$, that is bounded $ \|T(t)\|\leq  \delta < 1 $. Using the convergence of the von-Neumann series $ \sum_{k=0}^\infty T^k(t) = (\mathbbm{1}-T(t))^{-1} $, it is easy to see that $ \|U^{-1}(t)\|=\| [ \mathbbm{1} - T(t)]^{-1}\| $  is in turn $\| \sum_{k=0}^\infty T^k(t)\|$ which by triangle inequality again is less then or equal $ \sum_{k=0}^\infty \|T^k(t) \| = \sum_{k=0}^\infty \delta^k $, therefore 
\begin{equation}
    \|U^{-1}(t)\|  \leq \frac{1}{1-\delta}.
\end{equation}
Furthermore, 
\begin{equation}
 \begin{split}
    \|U^{-1}(t)- \mathbbm{1}\|&=\|[ \mathbbm{1}- U(t)]U^{-1}(t)\| \\
    &\leq\|\mathbbm{1}- U(t)\| \|U^{-1}(t)\| \\
    &\leq \frac{\delta}{1-\delta}.
 \end{split}
\end{equation}
We now concentrate on the leakage itself and note that:
\begin{align}
\|M(t)-M_{\text{eff}}(t)\| 
&= \| M(t) - U^{-1}(t)M(t)U(t_0)  \| \nonumber \\
&=\|{-[U^{-1}(t)- \mathbbm{1}]} M(t)-U^{-1}(t)M(t)[U(t_0)-\mathbbm{1}] \| \nonumber \\
&\leq \|U^{-1}(t)- \mathbbm{1}\| \| M(t)\|+ \|U^{-1}(t)\|\|M(t)\| \|U(t_0)-\mathbbm{1}\| \nonumber \\
&\leq \frac{\delta}{1-\delta} + \frac{\delta}{1-\delta}.
\end{align}
\begin{comment}
\begin{equation} 
    \begin{split}
\|M(t)-M_{\text{eff}}(t)\| 
&= \|U(t)M_{\text{eff}}(t)U^{-1}(t_0)-M_{\text{eff}}(t)\| \\
&= \|\big(U(t)-\mathbbm{1}\big)M_{\text{eff}}(t)U^{-1}(t_0)+M_{\text{eff}}(t)\big(U^{-1}(t_0)-\mathbbm{1}\big)\| \\
&\leq \|M_{\text{eff}}(t)\| \Big( \|U(t)-\mathbbm{1}\| \times \|U^{-1}(t_0)\| +\|U^{-1}(t_0)-\mathbbm{1}\|\Big) \\
&\leq \|U^{-1}(t)\| \times \|U(t_0)\| \times \Big( \|U(t)-\mathbbm{1}\| \times \|U^{-1}(t_0)\| +\|U^{-1}(t_0)-\mathbbm{1}\|\Big) \\
&\leq \frac{1}{1-\delta} \times (1+\delta) \times \Big( \delta \times \frac{1}{1-\delta}+\frac{\delta}{1-\delta}\Big),   
    \end{split}
\end{equation}
\end{comment}

In the third line, we used the triangle inequality, while the inequality in the fourth line follows directly from substituting the previous bounds in the last inequality.
Thus, the bound~\eqref{eq:leakage} is obtained.

\section{Solving the Bloch Equation with Radon's Lemma}\label{sec:appendixRadon}
We present a different approach to solving the Bloch equation~\eqref{eq:BE}: 
\begin{equation}\label{eq:rde}
    \dot{U}_k(t)=H(t)U_k(t)-U_k(t)H(t)U_k(t),
\end{equation}
subject to the Bloch condition~\eqref{eq:BC}: $P_kU_k(t)P_k=P_k$, and the definition of the wave operators $U_k(t)P_k=U_k(t)$, each acting on separate subspaces. To do that, we invoke Radon's Lemma~\cite[Theorem~3.1.1]{riccati-book}: 
\begin{lemma}
\emph{Part 1:} If $(X(t),Y(t))$ is a solution to the linear system of differential equations:
\begin{equation} \label{eq:radon-l}
\frac{d}{dt}
\begin{pmatrix}
X(t) \\
Y(t)
\end{pmatrix} =
\begin{pmatrix}
J_{1,1}(t) & J_{1,2}(t) \\
J_{2,1}(t) & J_{2,2}(t)
\end{pmatrix}
\begin{pmatrix}
X(t) \\
Y(t)
\end{pmatrix},
\end{equation}
where $X(t) \in \mathbbm{C} $ is non-singular on a domain, then $U(t)=Y(t)X\inv(t)$ is the solution to the Riccati differential equation (RDE) 
\begin{equation} \label{eq:radon-rd}
\dot{U}(t)=J_{2,1}(t)+J_{2,2}(t)U(t)-U(t)J_{1,1}(t)-U(t)J_{1,2}(t)U(t).
\end{equation}

\emph{Part 2:} If $U(t)$ is the solution to the RDE~\eqref{eq:radon-rd} and $X(t)$ solves $\dot{X}(t)=(J_{1,1}(t)+J_{1,2}(t)U(t))X(t)$, then $(X(t),Y(t))$ solves the linear system~\eqref{eq:radon-l}, where  $Y(t)=U(t)X(t)$.
\end{lemma}
\begin{proof}
\emph{Part 1:}
Differentiating $U(t)$ gives
\begin{equation}\label{eq:radon-proof1}
 \begin{split}
U(t)&=\dot{Y}(t)X\inv(t) - Y(t)X\inv(t)\dot{X}(t)X\inv(t) \\
&=J_{2,1}(t)X(t)X\inv(t) +J_{2,2}(t)Y(t)X\inv(t) \\ 
& \ \ \ {}-Y(t)X\inv(t)J_{1,1}(t)X(t)X\inv(t) - Y(t)X\inv(t)J_{1,2}(t)Y(t)X\inv(t) \\
&=J_{2,1}(t)+J_{2,2}(t)U(t)-U(t)J_{2,1}(t)-U(t)J_{2,1}(t)U(t). \ \  
\end{split}
\end{equation}

\emph{Part 2:}
Differentiating $(X(t),Y(t))$ follows that
\begin{equation} \label{eq:radon-proof2}
\begin{split}
\frac{d}{dt}
\begin{pmatrix}
X(t) \\
Y(t)
\end{pmatrix} &=
\begin{pmatrix}
\dot{X}(t) \\
\dot{U}(t)X(t)+U(t)\dot{X}(t)
\end{pmatrix} \\ &=
\begin{pmatrix}
(J_{1,1}(t)+J_{1,2}(t)U(t))X(t) \\
(J_{2,1}(t)+J_{2,2}(t)U(t))X(t)
\end{pmatrix} \\ &=
\begin{pmatrix}
J_{1,1}(t) & J_{1,2}(t) \\
J_{2,1}(t) & J_{2,2}(t)
\end{pmatrix}
\begin{pmatrix}
X(t) \\
Y(t)
\end{pmatrix}.  \ \  
\end{split}
\end{equation}
\end{proof}

It is worth noting that if the coefficients ${J_{i,j}(t)}$ are holomorphic on their domain, then the solutions $U(t)$ of the RDE are meromorphic. 

According to this, if $X_k(t)$ and $Y_k(t)$ are solutions for the associated linear system of differential equations:
\begin{equation} \label{eq:radon}
\frac{d}{dt}
\begin{pmatrix}
X_k(t) \\
Y_k(t)
\end{pmatrix} =
\begin{pmatrix}
0 & H(t) \\
0 & H(t)
\end{pmatrix}
\begin{pmatrix}
X_k(t) \\
Y_k(t)
\end{pmatrix},
\end{equation}
where $H(t)$ holomorphic in $\mathbb{C}$, then $U_k(t) = Y_k(t)X\inv_k(t)$ is a complex meromorphic solution for the Riccati differential equation~\eqref{eq:rde}, and the poles of $U_k(t)$ appear at most where $X_k(t)$ is not invertible. The real axis is the domain of interest for the time parameter $t \in [t_0,t_f]$.

In this equation~\eqref{eq:radon}, $Y_k(t)$ is decoupled, moreover $\frac{d}{dt} X_k(t) = \frac{d}{dt} Y(t)$, therefore the solution of this system of differential equations is:
\begin{gather}
    Y_k(t) = M(t) Y_k (t_0), \\
    X_k(t) = M(t)Y_k(t_0) + X_k(t_0) - Y_k(t_0), 
\end{gather}
with $M(t)$ as in Eq.~\eqref{eq:M}. Choosing any invertible initial condition $X_k(t_0)$ we have $Y_k(t_0) = U_k(t_0)X_k(t_0)$ and therefore 
\begin{gather}
    U_k(t) = M(t) U_k (t_0) X_k(t_0)X\inv_k(t), \\
    X_k(t) = \Pi_k(t)X_k(t_0), 
\end{gather}
with 
\begin{equation}\label{eq:pidef}
    \Pi_k(t) \equiv M(t)U_k(t_0) +\mathbbm{1} - U_k(t_0).
\end{equation}
\begin{proposition}
Let $\{P_k\}$ be a complete set of orthogonal projections, and let $U_k(t)=U_k(t)P_k$ satisfy the Bloch condition $P_kU_k(t)P_k=P_k$. Define $\Pi_k(t)$ as in Eq.~\eqref{eq:pidef}.  Then $\Pi_k(t)$ is block‐diagonal with respect to the decomposition $\{P_k\}$. Moreover 
\begin{equation}
    \Pi_k(t)=P_kM(t)U_k(t_0)P_k + (\mathbbm{1}-P_k).
\end{equation}
\end{proposition}

To prove this, first note that through the equality $U_k(t)(\mathbbm{1}-P_k)=0$ one gets
\begin{equation}
     \Pi_k(t)(\mathbbm{1}-P_k)= (\mathbbm{1}-P_k) =(\mathbbm{1}-P_k)^2 = (\mathbbm{1}-P_k) \Pi_k(t) (\mathbbm{1}-P_k).
\end{equation}
Similarity as $P_k U_k(t)=P_k$ the following holds:
\begin{equation}
    P_k \Pi_k(t) = P_kM(t)U_k(t_0) = P_k \Pi_k(t) P_k.
\end{equation}

It is easy to see, therefore, that as long as $P_kM(t)U_k(t_0)P_k$ is non-vanishing, the solution to the Bloch equation is given by 
\begin{equation}
    U_k(t) = M(t) U_k(t_0) X_k(t_0)X\inv_k(t_0)\Pi_k\inv(t) =M(t) U_k(t_0)P_k [P_kM(t)U_k(t_0)P_k + (\mathbbm{1}-P_k)]\inv ,
\end{equation}
which is the equivalent solution~\eqref{eq:BS} and equivalent existence condition as from the main text.

\section{Bounding the Unitary Transformation $V(t)$}\label{sec:appendix3}
Here we derive the bound~\eqref{eq:Vminus1} for the unitary transformation $V(t)$, given a uniformly bounded Bloch transformation such that $\|U(t)-\mathbbm{1}\| \leq \delta = \mathcal{O}(\gamma^{-1})$. This section is based on a similar calculation in Ref.~\cite{szabo_robust_2025}.

First, we show that by the triangle inequality,
\begin{align} \label{eq:appF3}
     \|U^\dag(t)U(t) - \openone\| &= \|U^\dag(t) (U(t)-\openone)+(U^\dag(t)- \openone)\| \nonumber \\
     &\leq  \|U^\dag(t)\| \|U(t)-\openone\|+\|U^\dag(t)- \openone\| \nonumber \\
     &\leq 2\delta+\delta^2,
\end{align}
where we used equation~\eqref{eq:apexU} from Appendix~\ref{sec:appendix}.
From the convergence of the von-Neumann series for any $\delta \ll 1$, one can also  identify
\begin{align}
    \|(U^\dag (t)U(t))^{-1/2}\|
    &=\|[\openone+(U^\dag (t)U(t)-\openone)]^{-1/2}\|
    \nonumber\\
    &\le\Bigl(1-\|U^\dag (t)U(t)-\openone\|\Bigr)^{-1/2}
    \nonumber\\
    &\le(1-2\delta-\delta^2)^{-1/2}.
\end{align}
and,
\begin{align}
    \|(U^\dag (t)U(t))^{-1/2}-\openone\|
    &=\|[\openone+(U^\dag (t)U(t)-\openone)]^{-1/2}-\openone\|
    \nonumber\\
    &\le\Bigl(1-\|U^\dag (t)U(t)-\openone\|\Bigr)^{-1/2}-1
    \nonumber\\
    &\le(1-2\delta-\delta^2)^{-1/2}-1.
\end{align}

Finally, we observe that
\begin{align}\label{eq:Wminus1}
        \|V(t)-\openone\| 
        &= \|U(t)(U^\dagger(t) U(t))^{-1/2}-\openone\|  \nonumber \\
        &= \|(U(t)-\openone)(U^\dagger(t) U(t))^{-1/2}+[(U^\dagger(t) U(t))^{-1/2}-\openone]\|  \nonumber \\
        &\leq \|U(t)-\openone\| \|(U^\dagger(t) U(t))^{-1/2}\|+ \|(U^\dagger(t) U(t))^{-1/2}-\openone\| \nonumber \\
        &
        \leq \frac{1+\delta}{\sqrt{1-2\delta-\delta^2}}-1,
\end{align}
the inequality~\eqref{eq:Vminus1} from the main text.

\end{widetext}

\bibliography{references}

\end{document}